\documentclass[a4paper]{article}
\pdfoutput=1

\usepackage[english]{babel}
\usepackage[utf8]{inputenc}
\usepackage{natbib}
\usepackage{amsmath}
\usepackage{amsfonts}
\usepackage{graphicx}
\usepackage{url}
\usepackage{listings}
\usepackage{array}
\usepackage{verbatim}
\usepackage{longtable}
\usepackage{booktabs}
\usepackage[official]{eurosym}
\usepackage{algorithm}
\usepackage[noend]{algpseudocode}
\usepackage{tikz}
\usetikzlibrary{shapes,arrows, shapes.geometric, shapes.symbols}
\newtheorem{definition}{Definition}
\newtheorem{proposition}{Proposition}
\newtheorem{proof}{Proof}
\newtheorem{axiom}{Axiom}
\usepackage{authblk}

\tikzstyle{every label}= [black]

\tikzstyle{place}=[circle,draw=black,minimum size=6mm, node distance=2cm]
\tikzstyle{transition}=[rectangle,draw=black, minimum size=6mm, node distance=2cm]
\tikzstyle{pre}=[<-,shorten <=1pt,>=stealth',semithick]
\tikzstyle{post}=[->,shorten >=1pt,>=stealth',semithick]

\tikzstyle{block} = [rectangle, draw, fill=gray!50, 
    text width=3cm, text centered, rounded corners, minimum height=4em]
\tikzstyle{line} =  [draw, thick, ->, shorten >=2pt] 
\tikzstyle{cloud} = [draw, ellipse,fill=red!20, node distance=3cm,
    minimum height=2em]
\tikzstyle{print} = [draw, tape, tape bend top=none, fill=gray!50, node distance=3cm
	 text width=5em, minimum height=4em, text centered]
\tikzstyle{decision}= [diamond, aspect=2, draw, fill=gray!50,
                        text width=2.5cm, text centered,
                        inner sep=1pt, rounded corners]

\begin{document}

\title{Shared value economics: an axiomatic approach}

\author[1]{Francisco Salas-Molina\footnote{Corresponding author. E-mail addresses: \textit{francisco.salas-molina@uv.es, jar@iiia.csic.es, filippo.bistafa@iiia.csic.es}}}
\author[2]{Juan A. Rodr\'iguez-Aguilar}
\author[2]{Filippo Bistaffa}

\affil[1]{Universitat de Val\`encia,  Av. Tarongers, 46022 Val\`encia, Spain}
\affil[2]{IIIA-CSIC, Campus UAB, 08913 Cerdanyola, Spain}

\bibliographystyle{apa}
\maketitle

\begin{abstract}
The concept of shared value was introduced by Porter and Kramer as a new conception of capitalism. Shared value describes the strategy of organizations that simultaneously enhance their competitiveness and the social conditions of related stakeholders such as employees, suppliers and the natural environment. The idea has generated strong interest, but also some controversy due to a lack of a precise definition, measurement techniques and difficulties to connect theory to practice. We overcome these drawbacks by proposing an economic framework based on three key aspects: coalition formation, sustainability and consistency, meaning that conclusions can be tested by means of logical deductions and empirical applications. The presence of multiple agents to create shared value and the optimization of both social and economic criteria in decision making represent the core of our quantitative definition of shared value. We also show how economic models can be characterized as shared value models by means of logical deductions. Summarizing, our proposal builds on the foundations of shared value to improve its understanding and to facilitate the suggestion of economic hypotheses, hence accommodating the concept of shared value within modern economic theory.
\\

\\


\end{abstract}


\section{Introduction\label{sec:intro}}

Michael E. Porter and Mark R. Kramer proposed the concept of shared value (SV) to clearly identify policies that an organization implements to simultaneously increase its competitiveness and enhance the social conditions of the communities in which they operate \citep{porter2011big}. The authors used the concept of SV for the first time in \cite{porter2006strategy} although it was fully developed later in \cite{porter2011big}. The underlying hypothesis is simple: focusing on both the economic and social aspects of the activities developed by organizations will benefit them in the long term. Here, social should be understood in broad sense including environmental issues. 


The idea of SV has received the attention of the management and economics research community. However, SV has also created some controversia. \cite{crane2014contesting} acknowledge the increasing popularity of the term but they also state that the concept of SV suffers from some serious shortcomings. Among others, the authors pointed out lack of originality, ignorance of the tensions between social and economic goals, and a naive conception of business compliance. \cite{porter2014response} responded to these criticisms by claiming that other authors such as \cite{elkington1994towards} and \cite{emerson2003blended} made important contributions to this area of research even though they focused on different aspects. Porter and Kramer further claimed that SV has gained so much traction because it aligns social progress with corporate self-interest in a concrete and highly tangible way. 

A recent literature review by \cite{dembek2016literature} attempted to answer the question: is shared value a theoretical concept or a buzzword? The authors concluded that, beyond the strong interest generated, SV conceptualization is vague, and it presents important discrepancies in the way it is defined and operationalized. It also overlaps with many other related concepts. By reviewing how SV was defined in the literature and what empirical instances were used to illustrate the concept, \cite{dembek2016literature} identified several key areas to develop the concept of SV from a theoretical perspective. Next, we transform these key areas in three main research questions: 

\begin{itemize}
    \item \textbf{RQ1}: Who should benefit from SV? 
    \item \textbf{RQ2}: How can we measure the outcomes of SV?  
    \item \textbf{RQ3}: How can we manage conflicts between social and economic goals? 
\end{itemize}

In order to address the previous research questions, we follow a theoretical and quantitative approach in a similar way to recent approaches to the related concept of circular economy \citep{korhonen2018circular,garcia2019defining}. Concepts, meanings and examples related to SV have been thoroughly discussed \citep{porter2006strategy,porter2011big,aakhus2012revisiting,crane2014contesting,dembek2016literature}. However, there is a need for an abstract architecture and a quantitative approach that enable researchers to guide the discussion and enrich further works on the topic. Even though SV is mainly an economic concept, current research do not approach SV from a quantitative perspective. The techniques to measure the benefits (economical and societal) that derive from SV policies are still missing. 

As a result, we here introduce and develop shared value economics (SVE). We provide a theoretical framework for SVE by integrating three distinct but complementary theories: (1) multiagent systems; (2) utility theory; and (3) multiple criteria decision making. Conveniently, we use multiagent systems to address RQ1, we rely on utility theory to answer RQ2, and we propose multiple criteria decision making to solve RQ3. Once we characterize shared value economics by means of an axiomatic approach, we address an additional interesting research question:
\begin{itemize}
    \item \textbf{RQ4}. Given an economic model, can we classify it as a SVE model or as as non-SVE model?
\end{itemize}

Summarizing, our theoretical framework builds on concept of shared value to improve its understanding and to facilitate the suggestion of economic hypotheses. This framework guides companies and external analysts in the study of shared value economics by providing the quantitative foundations of shared value economics. Our axiomatic approach allows to develop formal reasoning within the field of shared value economics. We describe such a formal reasoning by addressing the important question of characterizing existing economic models as shared value models by means of logical deductions. Finally, our proposal represents a tool to estimate the impact of different shared value policies in economic terms, hence accommodating the concept within modern economic theory.

In Section~\ref{sec:axiom}, we propose an axiomatic approach to the concept of shared value. Next, in Section~\ref{sec:found}, we develop the quantitative foundations of shared value economics. In Section~\ref{sec:char}, we illustrate how can we classify an economic model within the field of shared value economics. Finally, Section~\ref{sec:conc} concludes and provides natural extensions of our work.

\section{An axiomatic approach to share value\label{sec:axiom}}



This section develops the concept of SV following an axiomatic approach. Shared value was defined by \cite{porter2011big} as follows.

\begin{definition}
(\textbf{Shared value}). Policies and operating practices that enhance the competitiveness of a company while simultaneously advancing the economic and social conditions in the communities in which operates.  
\label{def:SV}
\end{definition}


Next, we make a critical assumption.
\\

\noindent
\textbf{Main assumption}. Shared value can be measured. 
\\

Since our axiomatic approach is intended to quantitatively develop the concept of SV, the previous assumption is a basic requirement. To this end, we rely on utility functions as a numerical representation of preference and satisfaction of agents within an economic context. From the main assumption, we next consider the following axioms or premises.
\\

\begin{axiom} 
\textbf{Coalition formation}. Shared value results from a coalition.\\
\label{ax1}
\end{axiom}

Handling interactions between multiple SV agents will usually require some degree of collaboration to achieve their respective goals. As a result, the formation of coalitions represents a fundamental requirement for SV. The rationale for such a statement is that collaboration is a necessary step to establish policies that ultimately lead to joint performance improvement. For instance, in a supply chain scenario, companies interested in reducing the environmental impact of production processes by replacing virgin raw materials with recycled materials are forced to collaborate with current or new suppliers \citep{carter2008framework}. In ridesharing, users collaborate with other users through social networks in order to minimize travel costs and reduce pollution \citep{agatz2012optimization,ostrovsky2018carpooling}.

Coalition formation aims to determine how much collective shared value can be obtained by forming a coalition \citep{shehory1995task,sandholm1999coalition,cerquides2013tutorial,bistaffa2017cooperative}. Formally, given a finite set of agents $\mathcal{A}:\{a_1, a_2, \ldots, a_n\}$, a coalition is a subset $S \subseteq A$ of agents. We can evaluate each coalition $S \in 2^\mathcal{A}$ by means of a characteristic function $v : 2^\mathcal{A} \rightarrow \mathbb{R}$, that maps each coalition  to its value. Then, a coalition structure ($CS$) is a partition of the set of agents into disjoint coalitions and the value of a 
$CS$ is assessed through the sum of the values of its coalitions  \citep{bistaffa2017cooperative}:
\begin{equation}
V(CS) = \sum_{S \in CS} v(S).   
\end{equation}

We here do not consider coalitions of size 1, also called singletons, in order to respect the essence of Definition \ref{def:SV}. However, the set of all possible coalitions excluding singletons is $2^{|\mathcal{A}|}-|\mathcal{A}|-1$. One way to address this issue is to impose restrictions on coalition formation by means of graph theory \citep{myerson1977graphs,demange2004group,voice2012coalition}. In this setting, an undirected graph $G=(\mathcal{A},\mathcal{L})$, is a set $\mathcal{A}$ of nodes representing agents and a set $\mathcal{L} \subseteq \mathcal{A} \times \mathcal{A}$ of links or edges establishing the relationships between the agents. As a result, links enable only connected agents to form a feasible coalition if its members represent the nodes of a connected subgraph of $G$ induced by $S$, i.e., for each pair of agents $a_i$ and $a_j$ belonging to $S$, there is a path in $G$ that connects $a_i$ and $a_j$ \citep{bistaffa2017cooperative}.

As an illustrative example, consider the main stakeholders of a company and its relationships as shown in Figure \ref{fig:graph}. A feasible coalition can be established between the company and its customers, between the company and its employees, but cannot be established between its customers and its employees. Considering now the company's supply chain, a coalition can be established between suppliers $S_1$, $S_3$, $S_4$ and $S_5$, but cannot be established between suppliers $S_2$ and $S_5$ because there is no link connecting suppliers $S_2$ and $S_5$.

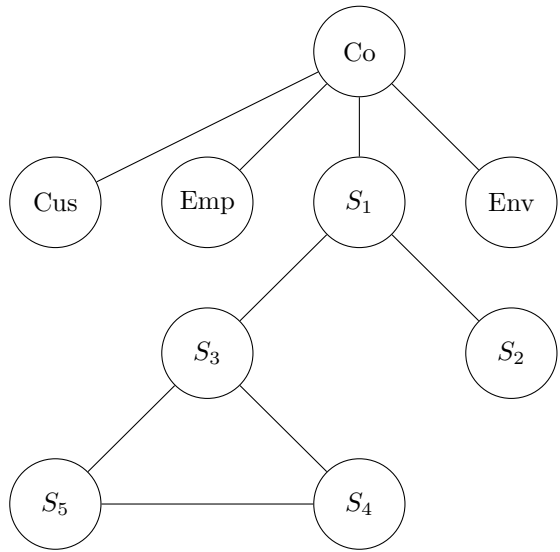
\begin{figure}[!htb]
\centering
\begin{tikzpicture}[node distance = 1cm]
\node[place] (level0) [minimum size=12mm]  {Co};
\node[place] (level1right) [below of = level0, minimum size=12mm]  {$S_1$}
	edge[] node[right] {}  (level0);
\node[place] (level1below) [right of = level1right,  minimum size=12mm]  {Env}
	edge[] node[right] {}  (level0);
\node[place] (level1left) [below of = level0, left of= level0, minimum size=12mm]  {Emp}
	edge[] node[left] {}  (level0);
\node[place] (level1Cus) [left of = level1left,  minimum size=12mm]  {Cus}
	edge[] node[right] {}  (level0);
\node[place] (level2right) [below of = level1right, right of= level1right, minimum size=12mm]  {$S_2$}
	edge[] node[right] {}  (level1right);
\node[place] (level2left) [below of = level1right, left of= level1right, minimum size=12mm]  {$S_3$}
	edge[] node[left] {}  (level1right);
\node[place] (level3right) [below of = level2left, right of= level2left, minimum size=12mm]  {$S_4$}
	edge[] node[right] {}  (level2left);
\node[place] (level3left) [below of = level2left, left of= level2left, minimum size=12mm]  {$S_5$}
	edge[] node[left] {}  (level2left)
	edge[] node[left] {}  (level3right);
\end{tikzpicture}
\caption{\label{fig:graph}A graph representing a company (Co) and its links to stakeholders such as customers(Cus), employees(Emp), suppliers ($S_1$ to $S_5$), and the natural environment (Env).}
\end{figure}

In the context of SVE, we assimilate coalition value calculation to SV assessment. Thus, we are interested in obtaining the most valuable coalition structure between intelligent agents.  However,  this evaluation cannot be only expressed in terms of economic criteria. If we want to adhere to what is stated in Definition~\ref{def:SV}, we need a second axiom.
\\



\begin{axiom} 
\textbf{Sustainability}. Shared value coalitions are sustainable.\label{ax2}
\\
\end{axiom}

At the core of the concept of SV, there is a need to consider not only the economic but also the social aspects of corporate policies. The term social is used by \cite{porter2011big} with a broad meaning that also includes environmental issues. From Definition \ref{def:SV}, we infer that economic and social criteria must be considered to make decisions in SVE. As a result, a second important feature of SVE is that it deals with problems located at the intersection of economic and social aspects. Thus, economic and social performance are the key criteria to elicit the best solutions as shown in Figure \ref{fig:sustainability}. 

Generally accepted accounting principles are used to refer to the standards and procedures that companies should follow to elaborate their financial statements. Similarly, there is a need to determine what policies generally constitute shared value. To solve this problem, we here propose the concept of sustainability. The intersection of economic and social (including environmental) aspects is usually coined as sustainability (see e.g. \cite{carter2008framework}). Then, a hypothetical project evaluated in terms of expected profits but also in terms of gender equality is, in essence, evaluated in terms of sustainability. Moreover, a new production process evaluated in terms of costs, workers safety and environmental impact is indeed evaluated in terms of sustainability. The idea that sustainability consists of three components (economics, society and the natural environment) was developed by \cite{elkington1994towards,elkington1998partnerships,elkington2004enter}. We here adapt this idea to develop a framework for SVE in which economic and social (including environmental) criteria are simultaneously considered, as stated in Definition~\ref{def:SV} and shown in Figure \ref{fig:sustainability}. In the context of SVE, we assume that sustainability is the combination of at least two criteria, one economic and one social, but maybe many others. As a result, we introduce the concept of sustainable coalition.

\begin{definition} 
(\textbf{Sustainable coalition}). Given a set of agents $\mathcal{A}$, a coalition $S \in 2^{\mathcal{A}}$ is said to be sustainable when its characteristic function $v: 2^\mathcal{A} \rightarrow \mathbb{R}$ is of the form:
\begin{equation}
v(S)=g(u_e,u_s)    
\end{equation}
where $u_e \in \mathcal{U}$ denotes an economic utility function and $u_s \in \mathcal{U}$ denotes a social utility function.
\label{def:sustcoal}
\end{definition}

The characterization of a utility function as sustainable may be problematic since we are trying to derive a mathematical expression from a multifaceted concept. Remarkable efforts have been made \citep{sachs2012millennium,griggs2013policy,un2014sdg} to propose a set of goals that considers economic, social and environmental dimensions to improve people’s lives and protect the planet for future generations. In the context of SVE, we propose the definition of a sustainable utility function as the combination of economic and social utility functions within a coalition. The rationale behind this proposal is twofold: first, SV is an economic and social concept; and second, the use of utility functions in economics and sociology is a well established field of research. In SVE, however, we restrict the domain of social utility functions to those that directly affect the economic and social conditions of an organization and its stakeholders (see Definition \ref{def:SV}).

\begin{figure}[!htb]
	\centering 
	\includegraphics[scale=0.4]{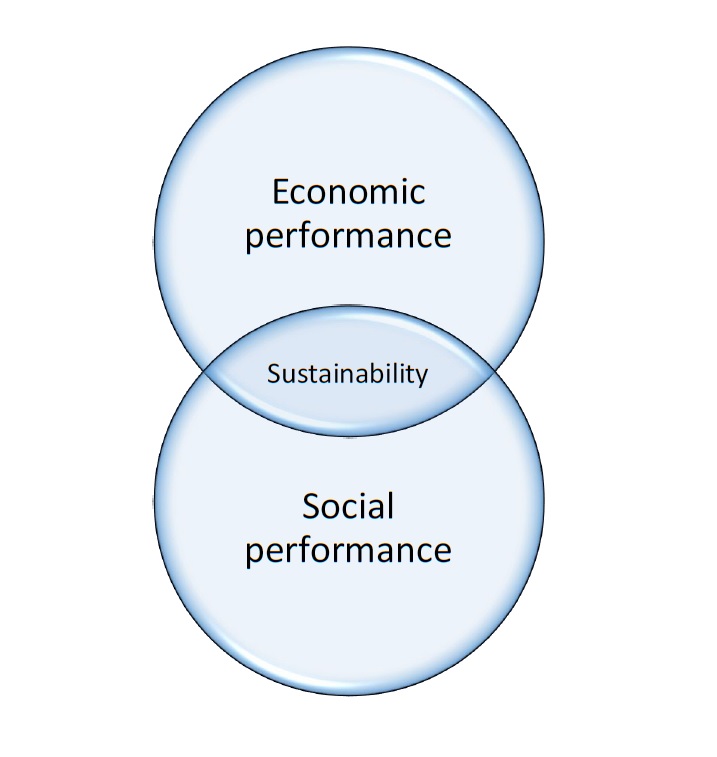}
    \caption{A graphical definition of sustainability.}
	\label{fig:sustainability}
\end{figure}

To illustrate Definition \ref{def:sustcoal}, we rely on the following examples. Indeed, this definition implies the presence of at least one economic utility function (we are dealing with economics) and at least one social utility function (we are dealing with shared value). \cite{porter2011big} refer to food companies refocusing from better taste and quantity to increase consumption to better nutrition. Here consumption is the economic utility function for the organization and nutrition is the social utility function for customers. The authors also report distribution companies redesigning packaging and rerouting tracks to reduce environmental impact. In this case, the sum of packaging and routing costs is the economic utility function and the environmental impact is the social utility function.

Axioms \ref{ax1} and \ref{ax2} characterize SVE by establishing its dimension and scope. However, a method to elicit the best solutions in terms of SV is still missing. What happens when economic, social and environmental criteria are in conflict? We need to look for a compromise solution from a SV perspective. We need a method to make decisions. The following axiom characterizes the way in which SV agents make decisions from a rational perspective.
\\

\begin{axiom}
\textbf{Consistency}. Shared value is both internally and externally consistent.\label{ax3}
\\
\end{axiom} 

As any economic paradigm, the theory of SV aims to achieve both internal and external consistency. Following \cite{popper1959logic}, internal consistency means that statements and conclusions of a theory are drawn by means of logical deduction. In addition, external consistency means that the conclusions derived from it can be tested by means of empirical applications. In other words, that the theory should be useful both in theory and in practice. The internal consistency of the theory of SV is shown by \cite{porter2011big}. In this paper, we extend the consistency of the theory of shared value in its external side. To this end, we rely on multiple criteria decision making (MCDM) as a way to ensure external consistency. The reasons to follow this way of action are the next: 1) SV decisions requires a balance between multiple agents and multiple objectives; 2) SV agents within a complex environment may face problems with imperfect information. MCDM provides sound methods to deal with these two issues as we next further elaborate.

Even though best solutions in economics have been traditionally determined by single objective optimization, decision making with multiple criteria is the rule rather than the exception in many research areas including economics \citep{zeleny1982multiple,yu1985multiple,romero1991handbook,ballestero1998multiple,branke2008multiobjective,zopounidis2016multiple}. SVE is not one of these exceptions since we are considering economic and social criteria to make decisions. Then, MCDM becomes a necessary and useful tool. 

MCDM dates back to Benjamin Franklin, when he suggested to divide a sheet of paper by a line into two columns, writing pros on one column and cons on the other one \citep{maccrimmon1973overview}. The decision rule is simple: the column with more points shows the way to act. From the great variety of MCDM techniques, we here focus on those that accommodate well to SVE. To this end, we propose a third condition for SV agents, namely, they are rational in the sense that they aim to maximize their utility whatever it may be. The method used to combine agent's utility by means of a coalition becomes a critical point in our approach to SVE. The decisions made by individual agents may have a different impact in the optimization of the coalition value due to the presence of bargaining power. For instance, in the possible coalition between a big corporation and its workers, the bargaining power is on the organization side. On the contrary, in the possible coalition between the same corporation and the government, the bargaining power is on the stakeholder side due the government's capability of law enforcement. However, it is important to highlight that in this paper we take the perspective of an external analyst to study SVE.

The notion of utility was formally developed by \cite{von1953theory} to describe individual preferences when making decisions. The inherent complexity associated to multiple agents with multiple preferences results in the use of well-defined utility functions as a sound surrogate for individual and global preferences. For instance, profit maximization is used as a surrogate for utility optimization of the owners since they are entitled to a share of profits \citep{mas1995microeconomic}. In our context, we use functions that ultimately depends on possible states for each agent as a surrogate for utility. In Definition \ref{def:sustcoal}, we expressed the value of a sustainable coalition between agents as a function $v=g(u_e,u_s)$ including, at least, an economic criteria $u_e$ and a social criteria $u_s$, but maybe many others. Now, we go one step further by requiring agents to be rational in the sense that they look for solutions that maximize their individual utility. In order to propose the right MCDM technique to analyze SVE under Axioms \ref{ax1}, \ref{ax2} and \ref{ax3}, we consider two plausible scenarios.
\\

\noindent 
\textbf{Scenario 1}. SV agents do not have specific targets for their utility functions.
\\

A common approach to economics is utility optimization under a given set of constraints. Rational agents aim to maximize their utility functions without setting any specific target. The logic is simple: the higher utility, the better. When considering two (or more) conflicting goals, the well-known concept of Pareto efficiency provides the combination of achievement such that one goal cannot be improved without decreasing the achievement of the second one. In this context, Compromise Programming (CP), introduced by \cite{zeleny1973compromise} and \cite{yu1973class}, represents a suitable tool to obtain compromise solutions. CP is based on the concept of ideal point, usually infeasible, given by the optimum values for the criteria under analysis, and the Zeleny's axiom of choice \citep{zeleny1973compromise}: alternatives that are closer to the ideal are preferred to those that are further. Finally, a family of distance (loss) functions $L_h$ of parameter $h \in \left[1,2, \ldots, \infty\right]$ are used to select the best compromise solutions: 
\begin{equation}
L_h = \left[ \sum_{j=1}^n w_j^h (u_j^* - u_j)^h \right]^{(1/h)}
\label{eq:dist}
\end{equation}
where $u_j^*$ is maximum attainable value for the utility function $u_j$ of the $j$-th agent. In our context, we use weights $w_i$ as an expression of the bargaining power of agents. By varying parameter $h$ between 1 and $\infty$, we derive the compromise set. An important advantage of a CP approach to SVE is the possibility to explore solutions within the range of the maximum aggregated achievement of criteria (by minimizing Manhattan distances $L_1$) and the maximum balance balance among criteria (by minimizing Chebyshev distances $L_{\infty}$). 

Within economics research, \cite{ballestero1991theorem,ballestero1993weighting,ballestero1994utility} show that this compromise set can be interpreted as a landing area (a subset of the Pareto efficient frontier) for utility maximization under plausible restrictions. Thus, CP represents a good surrogate of the utility maximum for rational agents. This argument is corroborated by subsequent works on macroeconomics \citep{andre2008using,andre2009defining}, microeconomics \citep{ballestero1998multiple,ballestero2000project,ballestero2015compromise}, and portfolio selection \citep{ballestero1998approximating,ballestero2004selecting}. In the context of SVE, we rely on CP to study compromise solutions from an external analyst point of view when social and economic criteria are in conflict, i.e., when economic and social achievement is out of the ``sweet-spot". 


Even though agents in Scenario 1 do not have specific targets for their utility functions, our CP setting usually requires that maximum (or anchor) and minimum (or nadir) values can be independently established for each of the utility functions (criteria) used by agents. These values allow to derive normalized indexes avoiding meaningless comparisons between criteria when utility functions are not homogeneous (different units or scales). In our CP approach to SVE, agents are able to forecast and compute expectations of their utility functions. Agents are then considered as artificially intelligent agents in the sense of \cite{marimon1990money}. This approach is related to the rational expectations theory introduced by \cite{muth1961rational} and later developed by \cite{sargent1975rational,sargent1976rational,lucas1981rational} and many others. Rational expectations theory describes economic situations in which outcomes do not systematically differ from what people expect to happen. The underlying assumption is that an agent's rationality is expressed in its behavior to maximize utility. This paradigm clearly departs from the Simonian bounded rationality which we consider in the Scenario 2.
\\

\noindent 
\textbf{Scenario 2}. SV agents have satisficing targets for their utility functions.
\\

Instead of assuming that economic agents are perfectly rational as in Scenario~1, we can reasonable assume that they face problems with imperfect information. The current complexity of the economic environment may make very difficult the maximization of their goals \citep{ballestero2000project}. An alternative way of action is following a satisificing logic in the Simonian sense of bounded rationality \citep{simon1955models,simon1979rational}. Bounded rationality suggests a different approach to rational expectations theory. For an interesting comparison between both theories see \cite{sent1997sargent}. Satisficing methods aim to find solutions that provide sufficiently satisfactory levels of goals by means of Goal Programming (GP). GP is probably the most widely used technique to put into practice this satisficing logic. GP is a multiobjective decision support tool that aims to minimize the sum of deviations between goal achievement and the aspiration levels (or targets) of the goals \citep{charnes1977goal,romero1991handbook,tamiz1998goal,jones2010practical}. By minimizing an objective function with the sum of deviations with respect to the satisfactory aspiration levels $b_j$ for each agent, we say that GP follows a satisfying logic. A weighted GP model can be expressed as follows:
\begin{equation}
\operatorname{min} \sum_{j=1}^n \left[w_j (\delta_j^+ +  \delta_j^-)  / b_j   \right] 
\label{eq:gp_obj}
\end{equation}
subject to:
\begin{equation}
u_j =  b_j + \delta_j^+ - \delta_j^-  \hspace{3mm} j =1,2, \ldots, n
\end{equation}
where $u_j$ is the utility function for the $j$-th agent, and $\delta_j^+,  \delta_j^- \geq 0$ are positive and negative deviations of the $j$-th utility with respect to its target $b_j \neq 0$. As in Scenario 1, we follow the approach of using weights $w_j$ adding up to one as an expression of the bargaining power of each agent. 


\section{Quantitative foundations of shared value economics\label{sec:found}}


Once we have presented our axiomatic approach, we are in a position to develop SVE by quantitatively defining SV concepts. To this end, we first rely on single output resource allocations as a common tool in economics \citep{mas1995microeconomic}. 



\begin{definition}
(\textbf{Resource allocation}). Given an economy with $n$ agents and $q$ goods, a resource allocation for the $j$-th agent is a vector $\boldsymbol{y}_j \in \mathbb{R}^q$ that describes the non-negative quantity $x_i$, with $i=1, \ldots, q-1$, of the resources required for an agent to obtain $y_q$ units of a desired output. 
\begin{equation}
\boldsymbol{y}_j = (-x_1,-x_2, \ldots, -x_{q-1}, y_q)  
\end{equation}
\end{definition}

Note that the use of some inputs may be zero. An example of a resource allocation for an industrial company could be $(-7,-12,-6,10)$, when they use 7 units of raw materials, 12 hours of machinery, and 6 hours of worker-time to produce 10 units of a final product. For an employee of this company, a resource allocation could be $(0,0,-1,1)$, when she/he dedicates one hour of her/him time and knowledge to deliver one hour of work to the company.


\cite{porter2011big} define value as benefits relative to costs. From allocation $\boldsymbol{y}_i$, we can straightforwardly derive a measure of value in economic units from outcomes and resources. Let us consider price vector $\boldsymbol{p}_j \in \mathbb{R}^q$, where each element $p_{jk}$ is the price paid by agent $j$ for resource $k$ when $k < q$, or obtained for output $k$ when $k=q$. Price vectors must be different for each agent since their produced outputs are also different. Then, we can compute the value $v_j$ obtained by agent $j$ as follows:
\begin{equation}
v_j = \boldsymbol{p}_j^T \cdot \boldsymbol{y}_j
\end{equation}
where $T$ denotes vector transposition. Here, we are interested in multiple allocations for a set of agents (see Axiom 1). A special case of multiple resource allocations, which we call coalition formation in the context of SVE, is a combination of $n$ resource allocations $(\boldsymbol{y}_1, \ldots, \boldsymbol{y}_n)$, for $n$ different agents with $n \geq 2$. Then, the value $v(S)$ of a coalition $S \subseteq \mathcal{A}$ of $n$ agents can be expressed as follows:
\begin{equation}
v(S) = \sum_{j=1}^n v_j(S)
\label{eq:covalue}
\end{equation}

We here interpret the utility derived from coalition $S$ as the value that agent~$j$ can achieve by forming coalition $S$:
\begin{equation}
    u_j(S):=v_j(S).
\end{equation}

From the set of all possible coalitions, we are interested in the set of feasible utility outcomes such that no other coalition can improve the payoffs of all its members. This fact ensures that agents have an incentive to form the coalition, hence leading to the game-theoretic concept of the core \citep{shapley1971cores,mas1995microeconomic,leyton2008essentials}. 

\begin{definition}
(\textbf{The core}). A payoff vector $\boldsymbol{u} \in \mathbb{R}^n$ is in the core of a transferable utility coalitional game $(\mathcal{A},v)$, where $\mathcal{A}$ denotes the set of agents and $v$ is the characteristic function, if there is no other coalition $S$ such that $u_j(S) \geq u_j(\mathcal{A})$ for all agents $j=1, \ldots, n$, and $u_j(S) > u_j(\mathcal{A})$ for some $j$. Formally, $\boldsymbol{u}$ is in the core of game $(\mathcal{A},v)$ if and only if:
\begin{equation}
    \sum_{j \in S} u_j \geq v(S) \hspace{3mm}\forall S \subseteq \mathcal{A}
\end{equation}
\end{definition}




The core requires that the sum of payoffs derived from any subcoalition must be at least as large as the amount that these agents could share if they formed a subcoalition. In other words, agents are better off if they are part of the grand coalition than if the form any subcoalition. From the wide range of coalitional games, we here focus on cooperative games with a non-empty core since, by definition, SV agents have an incentive to form the grand coalition with all agents involved. Then, the core payoff vector $\boldsymbol{u}$ satisfies:
\begin{equation}
    \sum_{j \in \mathcal{A}} u_j = v(\mathcal{A}).
\end{equation}

Furthermore, no sub-coalition $S$ has an incentive to break the grand coalition $\mathcal{A}$ since it would result worse-off. Next, we rely on the concept of core to define shared value games.

\begin{definition}
(\textbf{Shared value game}). A coalitional game $(\mathcal{A},v)$ is said to be a shared value game if and only if it is sustainable and it has a non-empty core.
\label{def:svgame}
\end{definition}


Axiom \ref{ax1} is implicit in Definition \ref{def:svgame} since a non-empty core requires coalition formation as a basic assumption in game theory. Sustainability is required to respect Axiom \ref{ax2} as proposed in Definition \ref{def:sustcoal}. Axiom \ref{ax3} is ensured by both the internal and external consistency of game theory. The rationale behind the use of the concept of the core to define a SV game is the fact that all agents in a shared value game have a strong motivation to cooperate. According to Definition \ref{def:svgame}, we are particularly interested in the type of games with a non-empty core disregarding the fact that, in general, the core is hard to compute. However, it is shown elsewhere \citep{shapley1971cores} that every convex game has a non-empty core.

\begin{definition}
(\textbf{Convex game}). A coalitional game $(\mathcal{A},v)$ is convex if the following inequality holds:
\begin{equation}
v(S \cup T) \geq v(S) + v(T) - v(S \cap T) \hspace{3mm} \forall S,T \subset \mathcal{A}.
\end{equation}
\end{definition}

To illustrate the concept of convexity, let us consider the possible coalition between the employees, denoted by $a_1$, of a law firm, denoted by $a_2$. The company have access to a number of customers interested in its services. However, it has not the knowledge and time required to serve its customers. This time and knowledge is provided by the firm's employees that, at the same time, have no access to the final customers. There are two available options to both lawyers and the company that are either to cooperate or not to cooperate. Assume that if they form a coalition to cooperate the payoffs are, respectively,  $u_1(\mathcal{A})=4$ and $u_2(\mathcal{A})=3$. If they do not cooperate the payoffs are zero to both agents. As a result, the coalition between the firm and its employees has value $v(\{a_1\} \cup \{a_2\})=7$, according to equation \eqref{eq:covalue}, and there is no incentive to break the coalition in a convex game since $v(\{a_1\}) + v(\{a_2\})$ would be zero and $\{a_1\} \cap \{a_2\} = \emptyset$.

Even though the existence of a non-empty core is a desirable property it does not ensure that a SV coalition is in any sense equitable. For example, the extreme points of a two-dimensional utility frontier $z(u_1,u_2)$ maximize the utility of one agent and simultaneously minimize the utility of the other agent as shown in Figure \ref{fig:equality}. It is then necessary to look for a compromise solution according to Axiom \ref{ax3}. Let us assume that women in our law firm example are paid less than men for the same job and responsibilities. All employees, represented by agent $a_1$, are interested in maximizing gender equality $u_1 = E$. For simplicity, we also assume that there are as many women as men in the company. Then, if $C_w$ is the total amount paid to woman and $C_m$ is the amount paid to men within a given period, we can express equality as follows:
\begin{equation}
    C_w = E \cdot C_m.
\label{eq:cost}
\end{equation}

The relationship between gender equality and wages is then quite simple. When $E=1$, woman are paid equal to men. When $E < 1$ women are paid less than men. The company, represented by agent $a_2$, is interested in maximizing benefits $u_2(S)=B$, which we can express as a real valued function of sales $Q \geq 0$ and human resources costs $C_w$ and $C_m$:
\begin{equation}
    B = Q - C_w - C_m.
\label{eq:benef}
\end{equation}

By merging Equations \eqref{eq:cost} and \eqref{eq:benef}, we can express benefits $B$ in terms of equality $E$:
\begin{equation}
    B = Q - E \cdot C_m - C_m = Q - C_m (1 + E).
\label{eq:EB1}
\end{equation}

Equation \eqref{eq:EB1} is the utility trade-off $z(E,B)$ in a normalized Equality-Benefits (E-B) space as shown in Figure \ref{fig:equality}. Normalization is achieved by means of two indexes $\theta_1$ and $\theta_2$ ranging in the interval $\left[0,1\right]$. Index $\theta_1 = E$ since the domain of equality is already $\left[0,1\right]$. Index $\theta_2$ is computed as follows:
\begin{equation}
    \theta_2 = \frac{B - B_*}{B^*-B_*}
    \label{eq:theta2}
\end{equation}
where $B^*=Q- C_m$ is the maximum benefit obtained when $E=0$, and $B_*=Q-2C_m$ is the minimum benefit obtained when $E=1$ in Equation~\eqref{eq:EB1}. Basic algebra leads to the normalized utility $z(E,B)=\theta_2 = 1-\theta_1$. 

\begin{figure}[htb]
\centering
\includegraphics[width=0.8\textwidth]{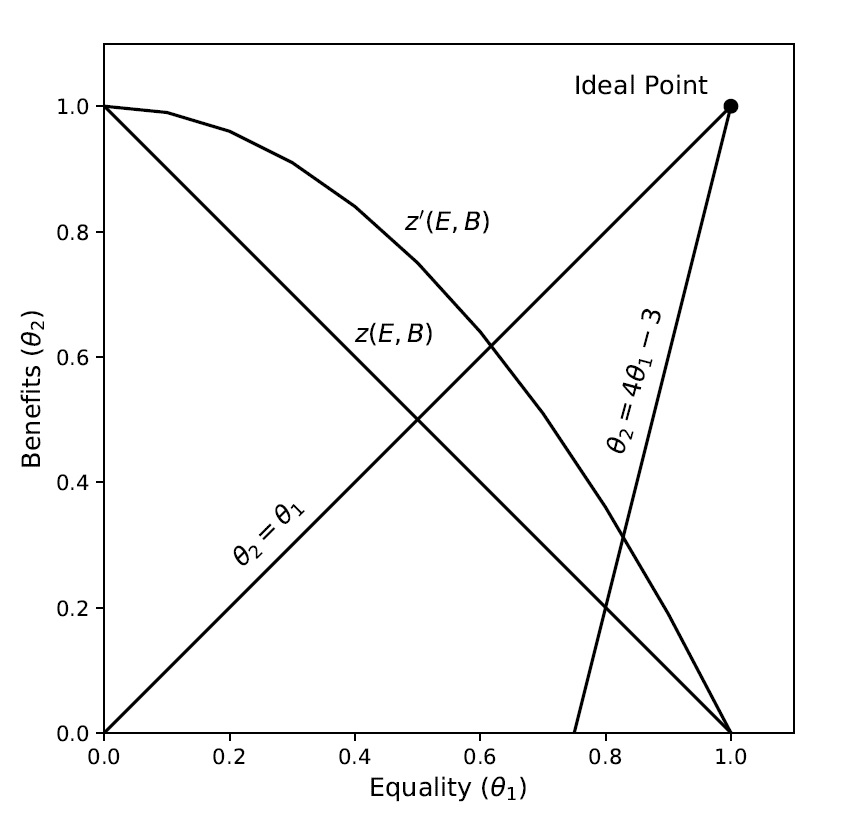}
\caption{\label{fig:equality}Creation of shared value.}
\end{figure}

The trade-off between benefits and equality is clear. The higher the equality, the lower the benefits. However, it is likely that an increase in equality will result in human resources productivity gains, hence improving benefits in the long-term by reducing costs. Let us assume that there is a training program for employees that ensures productivity gains. However, women are not willing to attend the program unless gender inequality is reduced. Then, cooperation (a coalition) between the law firm and its employees is a good strategy. Let us further consider that the coalition formation ensures a cost reduction proportional to the square of equality so that higher equality achievements can be obtained for the same desired level of benefits. In order to include productivity gains in our analysis, we can alternatively express benefits as follows:
\begin{equation}
    B' = Q - C_m - C_m E^2 = Q -  C_m (1 + E^2).
\label{eq:EB2}
\end{equation}

After normalizing $B'$ using equation \eqref{eq:theta2}, we obtain a new combined utility function $z'(E,B)=\theta_2 = 1 - \theta_1^2$. We observe in Figure \ref{fig:equality} a displacement of $z'(E,B)$ with respect to $z(E,B)$ towards the ideal equality-benefits combination, represented by the point $(1,1)$ in a normalized E-B space. Note also that the area under curve $z'(E,B)$ is higher than the area under curve $z(E,B)$. Adapting the recommendations in \cite{salas2017characterizing} to our context, we propose the following quantitative definition of shared value creation.

\begin{definition}
(\textbf{Shared value creation without targets}). Let $z(\theta_1,\theta_2)$ be the combined utility curve in a normalized space $\theta_1-\theta_2$ derived from a shared value game of two agents without specific targets where $\theta_1$ and $\theta_2$ are, respectively, a social and an economic index ranging in $\left[0,1\right]$. Shared value creation (SVC) is the increase in the area under the curve (AUC) produced by a displacement from $z(\theta_1,\theta_2)$ to $z^\prime(\theta_1,\theta_2)$. SVC is formally expressed as:
\begin{equation}
SVC = AUC(z^\prime,\theta_1,\theta_2) - AUC(z,\theta_1,\theta_2) 
\end{equation}
where 
\begin{equation}
AUC(z^\prime,\theta_1,\theta_2)=\int_0^1 z^\prime(\theta_1,\theta_2)d\theta_1 > AUC(z,\theta_1,\theta_2)=\int_0^1 z(\theta_1,\theta_2)d\theta_1.
\end{equation}
\label{def:svc1}
\end{definition}

This definition also ensures that the best solution for $z'(E,B)$ is a better solution than the best for $z(E,B)$ in terms of parametric distance $\mathcal{L}_h$. We can also generalize Definition \ref{def:svc1} to coalitional games of the form $(\mathcal{A},v)$ by considering the hypervolume of dimension $n$ enclosed under the efficient utility frontier within an $n$-dimensional normalized space. In the case of agents of the type described in Scenario 2, the equivalent general definition of shared value creation is expressed as follows:

\begin{definition}
(\textbf{Shared value creation with targets}). Let $\boldsymbol{x} \in \mathbb{R}^m$ a vector of feasible decision variables and $g(\boldsymbol{x}) \in \mathbb{R}^n$, a general multiobjective function under a minimization context for a shared value game including social and economic goals. SVC with specific targets is the reduction obtained by some function $g^\prime(\boldsymbol{x})$ with respect to $g(\boldsymbol{x})$. SVC is then formally expressed as:
\begin{equation}
SVC = g(\boldsymbol{x}) - g^\prime(\boldsymbol{x})
\end{equation}
where
\begin{equation}
g(\boldsymbol{x}) > g^\prime(\boldsymbol{x}). 
\end{equation}
\label{def:svc2}
\end{definition}

In the context of bicriteria CP, a perfectly well-balanced solution is given by the intersection of the efficient frontier with the path $\theta_2=\theta_1$, which ensures an equitable combination of achievements as shown in Figure \ref{fig:equality}. However, in this particular coalition, we face an interesting paradox. We here describe the paradox of equality as the fact that perfectly well-balanced solutions do not result in maximum equality indexes. At most, medium equality indexes are possible depending on the efficient frontier. Furthermore, equality indexes that are close to the desirable value of one are only possible when enough bargaining power produce extremely imbalanced solutions. For instance, if we follow the approach of expressing bargaining power by means of weights, it can be shown (see e.g. \cite{ballestero1998multiple}) that a well balance solution is given by:
\begin{equation}
    w_1 (1- \theta_1) = w_2 (1- \theta_2).
    \label{eq:bal}
\end{equation}

In the example of Figure \ref{fig:equality}, assume that the relative bargaining power of employees with respect to the company is expressed by $w_1=0.8$ and $w_2=0.2$. This expression implies that employees have four times more power than the law firm. As a result, the well balanced solution including bargaining power is given by the intersection of the line $\theta_2=4 \theta_1 - 3$, derived from equation~\eqref{eq:bal}, with the efficient frontier. Despite of the imbalance in the bargaining power, best compromise solutions are not able to produce maximum equality as the paradox of equality states. These results are consistent with the trade-off between equality and efficiency pointed out by \cite{okun2015equality}.

\section{Characterizing economic models\label{sec:char}}

Once we have described SVE from a quantitative point of view, we are in a position to analyze economic models proposed in the literature according to our axiomatic approach. As mentioned in the introduction, we here address the following research question: given an economic model, can we classify it as a SVE model? Here, we follow the approach of characterizing an economic model as a SVE model if the axioms introduced in Section \ref{sec:axiom} are ensured.

\subsection{A portfolio selection model}

Let us first consider the classical portfolio selection model proposed by \cite{markowitz1952portfolio}. An investor aims to obtain the best allocation of financial resources within a set of possible assets. The proposed model looks for solutions that simultaneously maximize returns and minimize risk. Since both returns and risk are objectives in conflict, there is a need to find a compromise solution. Let us further consider that a hypothetical investor is willing to include in the selection process social criteria as described in \cite{ballestero2012socially}. One may reasonably ask if this socially responsible investment model can be labelled as SVE. According to Axiom \ref{ax1}, the answer is no. Even though this model could perfectly fit under the sustainability and consistency axioms, the presence of only one agent, the investor, determines the characterization as a non-SVE model since the feature of coalition formation is not present in the model. A different characterization could take place if the managers of a socially responsible fund \citep{ballestero2012socially,ballestero2015socially} include in the model the particular preferences of the buyers of the fund. 

In the context of mutual funds, we can further consider the model proposed by \cite{ballestero2004selecting}. Individual investors ($a_i: \hspace{1mm} i=1,2 \ldots, n$) behave as customers buying the services of the fund rather than as decision-makers. The managers of the fund ($\overline{m}$) aim to differentiate from their competitors by offering thematic portfolios as a marketing strategy. Preferences of individual investors for profitability and risk are elicited by means of a questionnaire. The existence of this questionnaire is the way in which the coalition formation of Axiom \ref{ax1} between individual investors and the fund takes place. Assume also that sustainability in Axiom~\ref{ax2} is guaranteed by the presence of portfolios characterized by socially responsible assets. Finally, consistency in Axiom~\ref{ax3} is ensured by the CP method based described in \cite{ballestero2004selecting} that outputs a portfolio ranking derived from the questionnaire and the percentage of individual investors interested in different sets of thematic portfolios. The existence of a non-empty core is guaranteed by the following reasoning.

\begin{proposition}
Let $\mathcal{A}=\{a_1, a_2, \ldots, a_n, \overline{m}\}$ be a set formed by $n$ individual fund investors $a_i: \hspace{1mm} i=1,2 \ldots, n$, and a fund manager $\overline{m}$. The game $\mathcal({A},v)$ is a shared value game if $v$ is a function of the social and economic preferences of both investors and fund manager.
\end{proposition}

\begin{proof}
Assume that $\mathcal{A}$ is formed by subset $S=\{a_1, a_2, \ldots, a_{n}\}$ and subset $T=\{\overline{m}\}$. Investors in $S$ declare their collective social and economic preferences to the fund manager in $T$ by means of a questionnaire. The value of coalition $S \cup T$ is the sum of the utility for both subsets of agents: $v(S \cup T)=u(S)+u(T)$. Utility $u(S)$ for investors derived from buying the fund is assumed to be at lest as good as any other fund from a socially responsible investment point of view: $u(S) \geq v(S)$. It is also reasonable to accept that utility $u(T)$ for the fund manager is greater than value $v(T)=0$ derived from breaking coalition $S \cup T$. Sustainability is guaranteed by function $v$, which includes social and economic preferences. Inequality $v(S \cup T) \geq v(S) + v(T) - v(S \cap T)$ holds and the core is non-empty. As a result, we are dealing with a shared value game as proposed in Definition \ref{def:svgame}.
\end{proof}


\subsection{A carpooling model}

\cite{ostrovsky2018carpooling} have recently proposed a carpooling model to achieve socially efficient outcomes in a transportation marketplace with autonomous driving services. The main advantage of considering an economy in which all cars are self-driving is that carpooling with self-driving is a one-sided matching problem, while with regular cars it is a two-sided one, hence reducing its complexity. The model considers a finite set of riders $m=1,2, \ldots, M$ interested in collaborating to build trips, defined as feasible combinations of one or more riders over a finite set of road segments. 

Each rider has a non-negative valuation $v_m(t)$ for every trip $t$. Riders want to maximize the difference between value $v_m(t)$ and price $p_m$ paid for the trip. By assuming that all riders act a single agent for simplicity, the global utility for the entire set of riders is expressed as follows:
\begin{equation}
    U = \sum_{m=1}^M U_m(t_m,p_m) = \sum_{m=1}^M (v_m(t_m) - p_m).
    \label{eq:rider}
\end{equation}

A regulator owning road segments imposes the payment of reduced fees for tolls for each segment of the trip in which riders are involved. The presence of tolls acts as an incentive for riders, since it allows them to share the costs of those tolls. In addition, an assignment~$A$ is a set of trips such that each rider is involved in exactly one trip $t \in A$ with an associated cost $c(t)$. This assignment requires the fulfillment of the budget-balance condition stating that the sum of prices paid by the riders for their trips is greater than or equal to the sum of total physical costs of those trips for the riders and the total tolls on the road segments involved in those trips. Finally, the regulator plays the role of a social planner interested in maximizing a real valued social surplus $P$ defined as:
\begin{equation}
    P = \sum_{m=1}^M v_m(t_m) - \sum_{t\in A} c(t)
    \label{eq:planner}
\end{equation}

Summarizing, this carpooling model requires a coalition formation between riders, aiming at maximizing their transportation utility in equation \eqref{eq:rider}, and a regulator, aiming at maximizing the social surplus described in equation \eqref{eq:planner}. The presence of many agents, namely, riders and the regulator, respects Axiom \ref{ax1}. Since the regulator aims to maximize the social surplus of an assignment and riders want to maximize their utility, both the sustainability and consistency of Axioms \ref{ax2} and \ref{ax3} represent key elements of the model. Is the core of this carpooling model non-empty? To answer this question, we rely on the following proposition.

\begin{proposition}
Let $\mathcal{A}=\{a_1, a_2, \ldots, a_M, r\}$ be a set formed by $M$ riders $a_i: \hspace{1mm} i=1,2 \ldots, M$, and a regulator $r$. The game $\mathcal({A},v)$ is a shared value game if $v$ includes the social and economic needs of riders and regulator.
\end{proposition}
\begin{proof}
Assume that set $\mathcal{A}$ is formed by subset $S=\{a_1, a_2, \ldots, a_M\}$ and a subset $T=\{r\}$. Without the regulator, roads and reduced prices would not be possible resulting in $v(S) \leq U$. Without riders, social surplus $P=v(T)$ would be zero, hence preventing sustainibility. As a result, inequality $U + P \geq v(S) + v(T) - v(S \cap T)$ holds and the core is non-empty. Then, this model is a shared value game according to Definition \ref{def:svgame}.
\end{proof}


Thus, we conclude that the carpooling model presented by \cite{ostrovsky2018carpooling} is a is a SVE model. Note, however, that both utility $U$ and social surplus $P$ include the sum of riders' valuations. There is no trade-off between goals. The higher the utility, the higher the social surplus. Both agents have an incentive for cooperation, and we are dealing with the so-called ``sweet-spot" \citep{dembek2016literature}. To illustrate the common conflict between the objectives of economic agents, consider now that the regulator issues a license to a company that is interested in maximizing revenues $R$:
\begin{equation}
    R = \sum_{m=1}^M p_m.
    \label{eq:revenues}
\end{equation}

In this new context, SVE axioms hold too, but there is a trade-off between revenues and riders utility that can be computed by adding equations \eqref{eq:rider} and~\eqref{eq:revenues}:
\begin{equation}
    R = \sum_{m=1}^M v_m(t_m) - U.
\end{equation}

Due to this trade-off, it is necessary to look for a compromise solution. To this end, we first normalize utility and revenues, respectively, by means of indexes $\theta_1$ and $\theta_2$ following the same procedure that we use in equation \eqref{eq:theta2}. Then, we can depict normalized revenues $\theta_2$ in terms of utility $\theta_1$ as shown in Figure~\ref{fig:carpooling}. The line $\theta_2 = 1-\theta_1 $ represents the current (the actual-world) situation in which the revenues of the licensee are inversely proportional to the riders' utility. 

\begin{figure}[htb]
\centering
\includegraphics[width=0.8\textwidth]{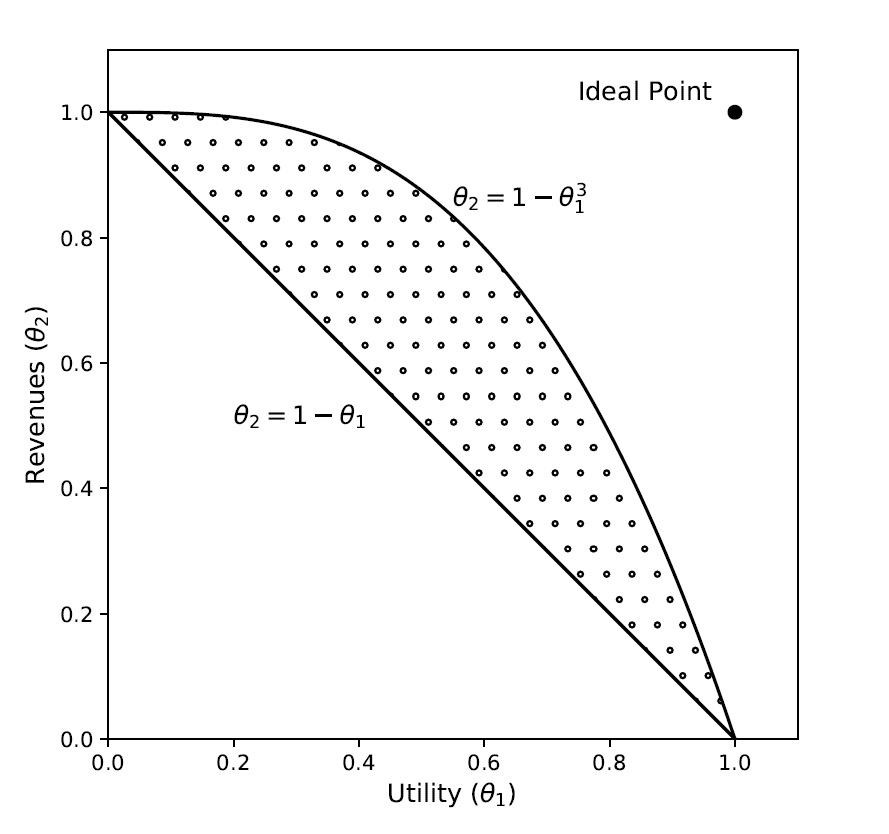}
\caption{\label{fig:carpooling}Creation of shared value in carpooling.}
\end{figure}

However, let us consider that the licensee is able to differentiate its carpooling service in such a way that riders' utility increases for the same price. This fact can be caused by a better matching algorithm, lower waiting times, better cars, or many other reasons. Let us further assume that the trade-off between utility and revenues shifts from $\theta_2 = 1-\theta_1$ to $\theta_2 = 1-\theta_1^3$ due to a better quality carpooling service. Each of the points of curve $\theta_2 = 1-\theta_1^3$ is a better solution than any of the points in the previous situation described by $\theta_2 = 1-\theta_1$. This displacement represents shared value creation as proposed in this paper. An additional interesting feature of our approach is that the opportunity of shared value creation can be estimated as the area comprised between these two curves.

\section{Conclusions\label{sec:conc}}

The notion of shared value raised doubts mainly about conceptualization, its beneficiaries and the way it is transformed into practical policies in modern economics. By means of an axiomatic approach, we show that shared value can be formally characterized as a first step to develop the concept and facilitate further research. By considering both the pioneering definition of shared value by Porter and Kramer and its criticisms, we propose a novel economic framework based on three axioms: coalition formation, sustainability and consistency. The presence of multiple agents that create value motivates the axiom of coalition formation to highlight the aspect that shared valued derives from a group formed by two or more agents. The axiom of sustainability ensures that both economic and social aspects are involved in shared value as its definition states. Finally, the axiom of consistency is a third key aspect to account for not only the theoretical principles of shared value but also for its practical implementation. In this sense, multiple criteria decision making methods represent a suitable approach to deal with shared value economics.

From the the axioms of coalition formation, sustainability and consistency, we derive the quantitative foundations of the theory of shared value. As a second result, we define shared value economics from a quantitative perspective as the coalition formation between multiple agents that make the best decisions under a long-term perspective of sustainability. Furthermore, we use utility functions to measure the outcomes of shared value and the game-theoretic concept of the core to formally define a shared value game and shared value creation. In order to obtain shared value policies in practice, multiple criteria decision making techniques provide a way to handle conflicting agents' goals as a surrogate for classical utility maximization approaches. A further advantage of considering multiple criteria decision making tools to deal with shared value economic problems is its ability to incorporate the concept of bounded rationality to cover a wider range of shared value agents.


Summarizing, our proposal contributes to consolidate the concept of shared value within modern economic theory. Our axiomatic approach builds on the foundations of shared value to improve its understanding and to facilitate formal reasoning and the suggestion of economic hypothesis. One of these hypotheses is the verification if an economic model can be classified as shared value economic model or not. As an illustration of such a formal reasoning, we show how a portfolio selection model and a carpooling model can be characterized by means of logical deductions. In this sense, determining how economic models can be adapted to fit the axioms of shared value economics is an interesting future line of work.

\bibliography{biblio}

\end{document}